\documentclass{llncs}

\usepackage{amsmath, amssymb}
\usepackage[english]{babel} % Why?

\usepackage{graphicx}
\usepackage{booktabs}

\usepackage[norelsize,ruled,vlined]{algorithm2e}
\newcommand{\spi}{\ensuremath{\sigma\pi^{-1}}}

\newcommand{\img}{\mbox{im}}

\title{On the Matrix Median Problem}
\author{Jo\~{a}o Paulo Pereira Zanetti\inst{1} \and Priscila Biller\inst{1} \and Jo\~{a}o Meidanis\inst{1,2}}
\date{}
\institute{Institute of Computing, University of Campinas, SP, Brazil \and
Scylla Bioinformatics, Campinas, Brazil}

\begin{document}

\maketitle

\begin{abstract} 
The Genome Median Problem is an important problem in phylogenetic reconstruction under rearrangement models. It can be stated as follows: given three genomes, find a fourth that minimizes the sum of the pairwise rearrangement distances between it and the three input genomes. Recently, Feij\~{a}o and Meidanis extended the algebraic theory for genome rearrangement to allow for linear chromosomes, thus yielding a new rearrangement model (the algebraic model), very close to the celebrated DCJ model. In this paper, we study the genome median problem under the algebraic model, whose complexity is currently open, proposing a more general form of the problem, the matrix median problem.
It is known that, for any metric distance, at least one of the corners
is a $\frac{4}{3}$-approximation of the median. Our results allow us to compute up to three additional matrix median candidates,
all of them with approximation ratios at least as good as the best corner, when the input
matrices come from genomes. From the application point of view, it is usually more interesting to locate medians farther from the corners.
We also show a fourth median candidate that gives better results in
cases we tried. However, we do not have proven bounds
for this fourth candidate yet.
\end{abstract}

% \begin{keywords}
% comparative genomics, genome rearrangements, algebraic distance
% \end{keywords}

\section{Introduction}

Genome rearrangements are evolutionary events where large, continuous pieces of the genome shuffle around, changing the order of genes in the genome of a species. Gene order data can be very useful in estimating the evolutionary distance between genomes, and also in reconstructing the gene order of ancestral genomes. The simplest form of inference of evolutionary scenarios based on gene order is the pairwise genome rearrangement problem: given two genomes, find a parsimonious rearrangement scenario between them, that is, the smallest sequence of rearrangement events that transforms one genome into the other.

For most rearrangement events proposed, this problem has already been solved, usually with linear or subquadratic algorithms. However, when more than two genomes are considered, inferring evolutionary scenarios becomes much more difficult. This problem is known as the \emph{multiple genome rearrangement problem} (MGRP):  given a set of genomes, find a tree with the given genomes as leaves and an assignment of genomes to the internal nodes such that the sum of all edge lengths (the pairwise distance between adjacent genomes) is minimal. 

The MGRP may still be hard even when only three genomes are considered, the so called \emph{genome median problem} (GMP): given three genomes, find a fourth genome that minimizes the sum of its pairwise distances to the other three. The GMP is NP-complete in most rearrangement models, with notable exceptions being the multichromosomal breakpoint distance~\cite{Tannier2009a} and the Single-Cut-or-Join (SCJ) model~\cite{Feijao2011}. The GMP is a particularly interesting problem, because several algorithms for the MGRP are based on repeatedly solving GMP instances, until convergence is reached (for instance, the pioneering BPAnalysis~\cite{Sankoff1998}, the more recent GRAPPA~\cite{Moret2001}, and MGR~\cite{Bourque2002}). 

In this paper we will focus on the algebraic rearrangement model proposed by Meidanis and Dias~\cite{MeidanisJoaoandDias2000}, recently extended to allow linear chromosomes in a very natural way by Feij\~{a}o and Meidanis~\cite{Feijao2012}. This extended algebraic rearrangement model is similar to the well-known Double-Cut-and-Join (DCJ) model~\cite{Yancopoulos2005a}, with a slight difference in the weight of single cut/join operations, where the weight for this operations is 1 in the DCJ model, but $1/2$ in the algebraic model. The algebraic pairwise distance problem can be solved in linear time, but the median problem remains open.

The main goal of this paper is to investigate the problem of computing
the algebraic median of three genomes. The median problem can be stated as follows: given three genomes
$\pi_1$, $\pi_2$, and $\pi_3$, and a distance metric $d$, find a genome $\mu$ that minimizes
$d(\mu; \pi_1, \pi_2, \pi_3)$, defined as the \emph{total score}
\[
  d(\mu; \pi_1, \pi_2, \pi_3) = d(\mu, \pi_1) + d(\mu, \pi_2) + d(\mu, \pi_3).
\]

We do not know yet the status of the median problem for the 
algebraic distance, but suspect it may be NP-hard as
well.  However, in this paper we show that viewing genomes (or even
general permutations) as matrices, the median can be approximated 
 quickly, although the matrix solution may not be always
translated back into permutations or genomes.  Nevertheless, this
positive result can help shed more light into the problem, by leading
to approximation solutions, or to special cases that can be solved
polynomially in the genome setting.

It is known that, for any distance satisfying the axioms of a metric, at least one of the corners
is a $\frac{4}{3}$-approximation of the median~\cite{Sankoff1996}. Our
results allow us to compute up to three additional matrix median candidates,
all of them with approximation ratios at least as good as the best corner, when the input
matrices come from genomes. Also, application-wise, it is usually more
interesting to locate medians farther from the corners~\cite{HS2012}.

The rest of this paper is organized as follows. In
Section~\ref{s:alg}, we have basic definitions regarding the
algebraic adjacency theory needed in this work. In
Section~\ref{s:res}, we show how genomes can also be seen as matrices,
define the matrix median problem, show our results, and propose an
algorithm. Finally, in Section~\ref{s:con}, we present our
conclusions.

\section{Algebraic Rearrangement Theory}
\label{s:alg}

We will start this section showing some basic definitions of the algebraic theory of Feij\~{a}o and Meidanis~\cite{Feijao2012}. 

\subsection{Permutations}

Given a set $E$, a \emph{permutation} $\alpha : E \rightarrow E$ is a bijective map from $E$ onto itself.  Permutations are
represented as parenthesized lists, with each element followed by its image. For instance, on $E=\{a,b,c\}$, $\alpha = (a\;b\;c)$ is the permutation that maps $a$ to $b$, $b$ to $c$, and maps $c$ back to $a$. This representation is not unique; $(b\;c\;a)$ and $(c\;a\;b)$
are equivalent. Permutations are composed of one or more \emph{cycles}. For instance, the permutation $\alpha = (a\;b\;c)(d\;e)(f)$ has three cycles. A cycle with $k$ elements is called a \emph{$k$-cycle}. An 1-cycle represents a fixed element in the permutation and is usually omitted.

The \emph{product} or \emph{composition} of two permutations $\alpha$, $\beta$ is denoted by $\alpha\beta$. The product $\alpha\beta$ is defined as $\alpha\beta(x)=\alpha(\beta(x))$ for $x\in E$. For instance, with $E=\{a ,b,c,d,e,f\}$, $\alpha=(b\;d\;e)$ and $\beta=(c\;a\;e\;b\;f\;d)$, we have $\alpha\beta = (c\;a\;b\;f\;e\;d)$. 

The \emph{identity permutation}, which maps every element into itself, will be denoted by $\mathbf{1}$. Every permutation $\alpha$ has an \emph{inverse} $\alpha^{-1}$ such that $\alpha\alpha^{-1} = \alpha^{-1}\alpha = \mathbf{1}$. For a cycle, the inverse is obtained by reverting the order of its elements: $(c\;b\;a)$ is the inverse of $(a\;b\;c)$.

A \emph{2-cycle decomposition} of a permutation $\alpha$ is a representation of $\alpha$ as a product of 2-cycles, not necessarily disjoint. All permutations have a 2-cycle decomposition. The \emph{norm} of a permutation $\alpha$, denoted by $\|\alpha \|$, is the minimum number of cycles in a 2-cycle decomposition of $\alpha$. For example, the permutation $\alpha = (a \; b \; c \; d)$ can be decomposed as $(a \; b)(b \; c)(c \; d)$, and $\|\alpha \| = 3$.

\subsection{Modeling Genomes as Permutations}

To model genomes with the Algebraic Theory, the formulation is similar to the set representation of a genome, used in several related works~\cite{Tannier2009a,Feijao2011}. In this representation, each gene $a$ has two \emph{extremities}, called \emph{tail} and \emph{head}, respectively denoted by $a_t$ and $a_h$, or alternatively using signs, where $-a = a_h$ and $+a = a_t$. An \emph{adjacency} is an unordered pair of extremities indicating a linkage between two consecutive genes in a chromosome. An extremity not adjacent to any other extremity in a genome is called a \emph{telomere}. A genome is represented as a set of adjacencies and telomeres (the telomeres may be omitted, when the gene set is given) where each extremity appears at most once.

According to algebraic
rearrangement theory~\cite{Feijao2012}, a genome can be seen as a
permutation $\pi : {E} \mapsto {E}$, where $E$ is the set of gene
extremities, with the added property that $\pi^2 = \mathbf{1}$, the
identity permutation.  In this paper, the \emph{distance} between 
two genomes or two permutations $\pi$ and $\sigma$ will be defined as
$d(\pi,\sigma) = \|\spi\|$. It is important 
to note that, in the original paper, the algebraic distance is defined 
as $\frac{\|\spi\|}{2}$. However, to avoid dealing with fractional 
numbers and to simplify the calculations, we will multiply the distances 
by 2 in this paper.

In the algebraic theory, genomes are represented by permutations, with a \emph{genome} being a product of 2-cycles, where each 2-cycle corresponds to an adjacency. Figure~\ref{fig:genome} shows an example of a genome and its representation as a permutation.

With these definitions, a distance between two genomes $\sigma$ and $\pi$ can be defined as $d(\pi,\sigma) = \|\spi\|$, as mentioned in the introduction.  The resulting distance agrees with DCJ in circular genomes, and in general these two distances are very close.

\begin{figure}[tp]
\centering
\includegraphics{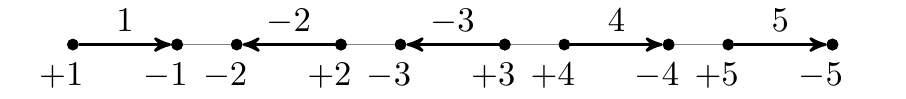}
\caption{A genome with one linear chromosome, represented by the permutation $\pi = ({-1}\;{-2})({+2}\;{-3})({+3}\;{+4})({-4}\;{+5})$.}
\label{fig:genome}
\end{figure}

\section{Results}
\label{s:res}

Permutations can be seen as matrices, and, in this section, we define
a matrix distance that corresponds to the algebraic distance and prove
that it is indeed a valid metric in general, that is, it applies to
all square matrices, not only those associated to permutations. Such
a metric can be useful in the computation of genome medians, and we show
here how to compute an approximate solution to the matrix median
problem, by solving a system of linear equations.

\subsection{Matrix Distance}
\label{s:dist}

Given two $n \times n$ matrices $A$ and $B$, we define the
\emph{distance} between them as:
\[
  d(A, B) = r(B - A),
\]
where $r(X)$ denotes the \textbf{rank} of matrix $X$. It is well-known that
\begin{equation}
\label{e:rank-dim-im}
  r(X) = \dim\img(X),
\end{equation}
where $\dim$ denotes the dimension of a vector space and $\img(X)$ is the \textbf{image} of $X$, namely, the space of all vectors that can be written as $Xv$ for some $v \in \mathbb{R}^n$.  Here we treat vectors as column matrices, that is, $n\times 1$ matrices.  Therefore, the distance satisfies
\[
  d(A,B) = \dim\img(B-A).
\]
for every pair of matrices $A$ and $B$.  It is also well-known that
\begin{equation}
\label{e:rank-compl-ker}
  r(X) = n - \dim \ker(X),
\end{equation}
for every $x\times n$ matrix $X$, where $\ker(X)$ is the \textbf{kernel} of $X$, defined as the space of all vectors $v \in \mathbb{R}^n$ such that $Xv=0$, so the distance satisfies yet another formula:
\[
  d(A,B) = n - \dim\ker(B-A).
\]

This distance can be shown to satisfy the conditions of a metric, that
is, it is symmetric, obeys the triangle inequality, and $d(A,B)=0$ if
and only if $A=B$~\cite{Delsarte1978}.
%\cite[p. 286]{encyclopedia}
%  The hardest part is the triangle inequality, which
% we show here.

% \begin{lemma}
% \label{l:trineq}
% For any three $n \times n$ matrices $A$, $B$, and $C$, we have
% \[
%   d(A,C) \leq d(A,B) + d(B,C).
% \]
% \end{lemma}

% \begin{proof}
% It is easy to see that
% \begin{eqnarray*}
%   \img(C - A) &\subseteq& \img(B - A) + \img(C - B),
% \end{eqnarray*}
% because any vector of the form $(C-A)v$ is clearly equal to the sum $(B-A)v + (C-B)v$ of a vector in $\img(B-A)$ with a vector in $\img(C-B)$.  

% Using Equation~\ref{e:rank-dim-im}, we conclude that
% \begin{eqnarray*}
%   d(A,C) &=& r(C - A) = \dim\img(C - A) \\
%          &\leq& \dim (\img (B - A) + \img (C - B))\\
%          &\leq& \dim \img (B - A) + \dim \img (C - B)\\
%          &=& r(B - A) + r(C - B)\\
%          &=& d(A,B) + d(B,C).
% \end{eqnarray*}\qed
% \end{proof}

Given this metric, the first interesting observation is that permutations (including
genomes) can be mapped to matrices in a distance-preserving
way.  Given a permutation $\alpha : {E} \mapsto {E}$, with
$|E| = n$, we first identify each element $v$ of $E$ with a unit
vector of $\mathbb{R}^n$, and then define $A$, the matrix counterpart of
$\alpha$, so that
\begin{equation}
  Av = \alpha v.\label{eq:dist}
\end{equation}
In Equation~\ref{eq:dist}, we use $v$ both as a unit vector of $\mathbb{R}^n$ in the left side,
and as an element of $E$ in the right side.

An example will help.  Let $\alpha$ be the permutation $(a\;b)(c\;d)$. Identify 
$a = \big[1\;0\;0\;0\big]^t$, $b = \big[0\;1\;0\;0\big]^t$, 
$c = \big[0\;0\;1\;0\big]^t$, and $d = \big[0\;0\;0\;1\big]^t$,
where $v^t$ is the transpose of vector $v$. We then have
\[
  A = \left[\begin{array}{cccc}
      0&1&0&0\\1&0&0&0\\0&0&0&1\\0&0&1&0
    \end{array}\right].
\]

It is well known that this mapping produces matrices $A$ that are
invertible, and that satisfy $A^{-1} = A^t$, where $A^t$ denotes the
transpose of matrix $A$.  Also, the identity permutation corresponds to the 
identity matrix $I$, and the product $\alpha\beta$ corresponds to
matrix $AB$, where $A$ is the matrix corresponding to $\alpha$, and
$B$ is the matrix corresponding to $\beta$.  If $\alpha$ happens to be
a genome, that is, if $\alpha^2 = \mathbf{1}$, then $A$ is a symmetric
matrix and vice-versa.

We now show that this mapping is distance-preserving.

\begin{lemma}
  For any permutation $\sigma$ and $\pi$, and their respective
  associated matrices $S$ and $P$, we have:

  \[
  d(\sigma,\pi) = d(S,P).
  \]
\end{lemma}

\begin{proof}

First, notice that it suffices to show that
\begin{equation}
\label{e:sole}
  \|\alpha\| = r(A - I),
\end{equation}
for any permutation $\alpha$ and associated matrix $A$.  Indeed, if $P$
is invertible then
\[
  r(S - P) = r(SP^{-1} - I),
\]
and then Equation~(\ref{e:sole}) will relate $d(\sigma, \pi)$ to the
distances of the corresponding matrices $S$ and $P$.

Then proceed to show that, for a $k$-cycle, Equation~(\ref{e:sole}) is true since
both sides are equal to $k-1$.  Finally, for a general permutation,
decompose it in disjoint cycles, and use the fact that, for
\emph{disjoint} permutations $\alpha$ and $\beta$, with associated
matrices $A$ and $B$, respectively, we have
\[
  \ker (A - I) \cap \ker (B - I) = \ker (AB - I),
\]
and
\[
  \ker (A - I) + \ker (B - I) = \mathbb{R}^n,
\]
which guarantee that
\[
  n - \dim \ker (AB - I) = n - \dim \ker (A - I) + n - \dim \ker (B - I).
\]
or
\[
  r(AB - I) = r(A - I) + r(B - I),
\]
because of Equation~\ref{e:rank-compl-ker}.

Therefore, if Equation~(\ref{e:sole}) is valid for $\alpha$ and $\beta$, and if $\alpha$ and $\beta$ are disjoint, then Equation~(\ref{e:sole}) is valid for the product $\alpha\beta$. Since any permutation can be written as a product of disjoint cycles, Equation~(\ref{e:sole}) is valid in general.
\qed
\end{proof}

\subsection{Matrix Median}
\label{s:mat-med}

Because the correspondence between permutations and matrices preserves
distances, it makes sense to study the matrix median problem as a way
of shedding light into the permutation median problem, which in turn is
related to the genome median problem.

Let $A$, $B$, and $C$ be three $n \times n$ matrices.  Suppose we want
to find a matrix $M$ such that
\[
  d(M; A, B, C) = d(M, A) + d(M, B) + d(M, C)
\]
is minimized.  In order to have small $d(M,A)$, $M$ must be equal to
$A$ in a large subspace, so that $\ker(A - M)$ is large.  Similarly
with $B$ and $C$.

This suggests the following strategy.
Decompose $\mathbb{R}^n$ as a direct sum of five subspaces, where
the following relations are true: 
(1) $A = B = C$,  (2) $A = B \neq C$, (3) $A \neq B = C$, 
(4) $A = C \neq B$, and (5) $A \neq B \neq C \neq A$. 

In the first subspace, since $A$, $B$, and
$C$ all have the same behaviour, $M$ should also do the same thing.
In the second subspace, since $A=B$ but $C$ is different, it is better
for $M$ to go with $A$ and $B$.  Likewise, in the third subspace $M$
should concur with $B$ and $C$, and with $A$ and $C$ in the fourth.
Finally, in the final subspace it seems hard to gain points in two
different distances, so the best course for $M$ would be to mimic one of
$A$, $B$, or $C$.

Therefore, making $M$ equal to $A$, except in the third subspace,
where it should be equal to $B$ (and $C$)
should yield a good approximation of a median, if not a median.  The rest of this section
will be devoted to showing the details on this construction.

\subsection{Partitioning $\mathbb{R}^n$}
\label{sec:part}

We begin by introducing notation aimed at formalizing
subspaces such as $A=B\neq C$.
Given $n\times n$ matrices $A$, $B$, and $C$,
we will use a dotted notation to indicate a partition,
e.g., $ .  AB  .  C . $ means a partition where
$A$ and $B$ are in one class, and $C$ is in another class
by itself.  To each such partition, we associate a vector
subspace of $\mathbb{R}^n$ formed by those vectors having the same
image in each class:
\[
  V( .  AB  .  C . ) = \{ v \in \mathbb{R}^n | Av = Bv \}.
\]

Notice that singleton classes do not impose additional restrictions.
With this notation, the subspace we used to call $A=B=C$ can be written
as $V( .  ABC . )$.  Notice also that $V( .  A .  B .  C . ) = \mathbb{R}^n$.

We need also a notation for strict subspaces, such as $A=B\neq C$,
where distinct classes actually disagree, that is, subspaces where
vectors have different images under each class.
There can be more than one
subspace satisfying this property, but we can use orthogonality to
define a unique subspace. For a partition $p$, we define $V_*(p)$ as
the orthogonal complement of the sum of the partitions strictly
refined by $p$ with respect to $V(p)$:
\[
V_*(p) = V(p) \cap (\sum_{p < q} V(q))^{\perp}
\]
where $p < q$ means that partition $p$ strictly 
refines partition $q$.  In other words, we want to capture
the part of the subspace $V(p)$ that is orthogonal to the sum of the
subspaces corresponding to coarser partitions.  Notice
that $p < q$ implies $V(q) \subseteq V(p)$.

It is easy to see that the $V_*$ subspaces are pairwise disjoint, but 
this is not enough to prove that their direct sum 
is $\mathbb{R}^n$. So, the proof will start from the basic sum
$V_*(.ABC.) \oplus V_*(.AB.C.)$, where we already know that $V_*(.ABC.) \cap V_*(.AB.C.) = \{0\}$,
and it will add one subspace in the sum at a time, ensuring that the intersection 
between the new subspace and the sum of the subspaces previously included  
contains the zero vector only.

\begin{lemma}\label{l:dirsum-first}
  If A, B, and C are permutation matrices, then
  \[(V_*(.ABC.)+V_*(.AB.C.)) \cap V_*(.BC.A.) = \{0\}.\]
\end{lemma}
\begin{proof}
  Let $v$ be a vector such that
  \[v \in (V_*(.ABC.) + V_*(.AB.C.)) \cap V_*(.BC.A.).\]

  We have that $Av = Bv$ and also $Av = Cv$. Therefore $Av = Bv = Cv$ and, consequently, $v = 0$, since $v \in V_*(.ABC.)$.
  \qed
\end{proof}

Before we add more subspaces to the sum, we will need the result
of the following lemma. 

\begin{lemma}
\label{l:permsum}
 If $A$, $B$, and $C$ are permutation matrices, and if \[2Bv = Av + Cv\] for a given vector $v$, then $Av=Cv$.
\end{lemma}
\begin{proof}
Denote by $|x|$ the norm of a vector $x$. Note that $A$, $B$, and $C$ preserve norms, thus
  \[|Av + Cv| = |2Bv| = 2|v| = |v| + |v| = |Av| + |Cv|.\]
But, if the norm of the sum is equal to the sum of the norms, then the two vectors are parallel, 
and have the same orientation. In other words, there is a positive scalar $c$ such that $cAv = Cv$. But we
already have that $|Av|=|v|=|Cv|$. Therefore, $c=1$.
\qed
\end{proof}

\begin{lemma}
 If A, B, and C are permutation matrices, then
  \[(V_*(.ABC.)+V_*(.AB.C.)+V_*(.BC.A.)) \cap V_*(.AC.B.) = \{0\}.\]
\end{lemma}
\begin{proof}
Suppose that $u + v + w \in V_*(.AC.B.)$, where $u \in V_*(.ABC.)$, 
$v \in V_*(.AB.C.)$, and $w \in V_*(.BC.A.)$. We have that
$A(u+v+w) = C(u+v+w)$, which implies $A(v+w) = C(v+w)$, since $Au=Cu$.
Thus, we have
\begin{align*}
Av + Aw &= Cv + Cw\\
Bv + Aw &= Cv + Bw\\
B(v-w) &= Cv - Aw.
\end{align*}
Now we will apply $A$ and $C$ to $v-w$, and sum the results, obtaining:
\begin{align*}
A(v-w) &= Bv - Aw \\
C(v-w) &= Cv - Bw \\
A(v-w) + C(v-w) &= B(v-w) + Cv - Aw = 2B(v-w) 
\end{align*}
By Lemma~\ref{l:permsum}, we conclude that $A(v-w) = C(v-w)$. But we also have 
that $A(v+w) = C(v+w)$, which implies $Av = Cv$ and $Aw = Cw$. In other words, $u + v + w \in V_*(.ABC.)$.
But $u+v+w$ also belongs to $V_*(.AC.B.)$. It follows that $u + v + w = 0$.
\qed
\end{proof}

\begin{lemma}\label{l:dirsum-last}
  If A, B, and C are permutation matrices, then
 \[(V_*(.ABC.)+V_*(.AB.C.)+V_*(.BC.A.)+V_*(.AC.B.)) \cap V_*(.A.B.C.) = \{0\}.\]  
\end{lemma}
\begin{proof}
By definition, the two subspaces whose intersection is indicated in the left-hand side are orthogonal to one another. Therefore, they have a zero intersection.
  \qed
\end{proof}

With Lemmas~\ref{l:dirsum-first} through~\ref{l:dirsum-last} we get to the following theorem:

\begin{theorem}
\label{t:dirsum}
 If A, B, and C are permutation matrices, then 
 \[ \mathbb{R}^n = V_*(.ABC.) \oplus V_*(.AB.C.) \oplus V_*(.BC.A.) \oplus V_*(.AC.B.) \oplus V_*(.A.B.C.). \]
\end{theorem}

It is important to observe that Theorem~\ref{t:dirsum} does not apply to 
general matrices, for instance:

\[
A =
\begin{bmatrix}
  0&0\\
  0&0
\end{bmatrix}
,
B =
\begin{bmatrix}
  0&1\\
  0&1
\end{bmatrix}
,
C =
\begin{bmatrix}
  1&0\\
  1&0
\end{bmatrix}
\]

With these three matrices, we have
\[V_*(.ABC.) = \{0\},\]
\[V_*(.AB.C.) = \left< [1 \; 0]^t \right>,\]
\[V_*(.BC.A.) = \left< [1 \; 0]^t, [0 \; 1]^t \right>,\]
\[V_*(.AC.B.) = \left< [0 \; 1]^t \right>,\]
where $\left< X \right>$ denotes the space spanned by the set $X$.

We can see that, in this case, \[(V_*(.ABC.)+V_*(.AB.C.)+V_*(.BC.A.)) \cap V_*(.AC.B.) \neq \{0\}.\]

\subsection{Computing Median Candidates}
\label{sec:comp}

We now get into further detail in the procedure to compute these
matrices. We saw that, when $A$, $B$, and $C$ are permutation matrices,
$\mathbb{R}^n$ can be decomposed into a direct sum of $V_*$
subspaces. We will now implement the procedure outlined in
Section~\ref{s:mat-med}, summarized in Table~\ref{tab:contrib}, to
obtain a median candidate $M_A$.

\begin{table}
  \centering
  \caption{Distance contribution --- Given three permutation matrices $A$, $B$ and $C$, this table shows the distance contribution of each of the five subspaces partitioning $\mathbb{R}^n$ to the distances $d(M_A,A)$, $d(M_A,B)$, and $d(M_A,C)$, for a candidate median matrix $M_A$.}
  \label{tab:contrib}
  \vspace{1em}
  \begin{tabular}{ccccc}
    \toprule
    & & \multicolumn{3}{c}{Contributes to \ldots}\\ \cmidrule{3-5}
    Subspace & $M_A = \ldots$ & $d(M,A)$ &  $d(M,B)$ &  $d(M,C)$ \\
    \midrule
    $V_*(.A.B.C.)$ & $A$ & no  & yes & yes \\
    $V_*(.A B.C.)$ & $A$ & no  & no  & yes \\
    $V_*(.B C.A.)$ & $B$ & yes & no  & no  \\
    $V_*(.A C.B.)$ & $A$ & no  & yes & no  \\
    $V_*(.A B C.)$ & $A$ & no  & no  & no \\
    \bottomrule
  \end{tabular}
\end{table}

One way to implement this strategy is to compute orthogonal projection
matrices $P_1$, $P_2$, $P_3$, $P_4$, and $P_5$ for each of the
subspaces and then compute $M_A$ as follows:
\begin{align*}
  M_A &= AP_1 + AP_2 + BP_3 + AP_4 + AP_5\\
      &= A + (B-A)P_3 \,,
\end{align*}
since $P_1 + P_2 + P_3 + P_4 + P_5 = I$. To obtain each matrix $P_i$
all we need is a $n \times k_i$ matrix $B_i$ whose columns form an
orthonormal basis of the corresponding subspace, because $P_i$ can
then be written as $B_iB_i^t$.

To build the $B_i$s, one possibility is to use Function~\ref{alg:add} below, a basic routine to
expand an orthonormal basis so that it can generate a given extra vector $v$. It
projects the new vector in the orthogonal complement of the subspace
generated by the original basis and then adds the normalized projection
to form the new basis.  Function Add's complexity is $O(kn)$ arithmetic
operations (additions, subtractions, multiplications, and divisions).

\begin{function}
  \DontPrintSemicolon
  \LinesNumbered
  
  \KwData{An orthonormal basis $B = \{v_1, \ldots, v_k\}$ and a vector $v$.}
  \KwResult{An augmented basis.}

  \BlankLine
  
  $w \leftarrow v - \sum_{i=1}^k v_iv^tv_i$\;
  \eIf{$w = 0$}{Return $B$\;}{
    Normalize $w$\;
    Return $B \cup \{w\}$\;
  }

  \caption{Add($B$, $v$) Augments an orthonormal basis $B$ so that it can generate the vector $v$.}
  \label{alg:add}
\end{function}

Algorithm~\ref{alg:bases} uses Function~\ref{alg:add} and a linear
system solving method Solve, that returns a linearly independent set
of vectors spanning the solution space, to determine orthonormal
bases for each of the five $V_*$ subspaces.  Assuming that we can solve
a system of $n$ linear equations in time $O(n^3)$, the total complexity
for Algorithm~\ref{alg:bases} is $O(n^3)$ as well.

\begin{algorithm}
  \DontPrintSemicolon
  \LinesNumbered
  
  \KwData{Three $n \times n$ permutation matrices $A$, $B$ and $C$.}
  \KwResult{The bases for the $V_*$ subspaces.}
  
  \BlankLine

  $B \leftarrow \emptyset$\;

  $L \leftarrow$ Solve $\{(A-B)v = 0$, $(A-C)v = 0\}$\;
  \ForEach{$v \in L$}{$B \leftarrow$ \Add{$B$,$v$}\;}
  $B_1 \leftarrow B$ \tcp{Basis for $V_*(.ABC.)$}

  $L \leftarrow$ Solve $\{(A-B)v = 0\}$\;
  \ForEach{$v \in L$}{$B \leftarrow$ \Add{$B$,$v$}\;}
  $B_2 \leftarrow B-B_1$ \tcp{Basis for $V_*(.AB.C.)$}

  $L \leftarrow$ Solve $\{(B-C)v = 0\}$\;
  \ForEach{$v \in L$}{$B \leftarrow$ \Add{$B$,$v$}\;}
  $B_3 \leftarrow B-B_1-B_2$ \tcp{Basis for $V_*(.BC.A.)$}

  $L \leftarrow$ Solve $\{(A-C)v = 0\}$\;
  \ForEach{$v \in L$}{$B \leftarrow$ \Add{$B$,$v$}\;}
  $B_4 \leftarrow B-B_1-B_2-B_3$ \tcp{Basis for $V_*(.AC.B.)$}

  $L \leftarrow$ canonical basis for $\mathbb{R}^n$\;
  \ForEach{$v \in L$}{$B \leftarrow$ \Add{$B$,$v$}\;}
  $B_5 \leftarrow B-B_1-B_2-B_3-B_4$ \tcp{Basis for $V_*(.A.B.C.)$}

  \caption{Computation of orthonormal bases for each of the $V_*$ subspaces.}
  \label{alg:bases}
\end{algorithm}

Once we have determined orthonormal bases for the $V_*$ subspaces, it
is easy to find projection matrices for them. For each $B_i$, its
projection matrix is $P_i = B_iB_i^t$, as previously said.

Knowing the projection matrices, we can then finally compute the median
candidates, in Algorithm~\ref{alg:med}.  Previously, we saw how to
compute $M_A$.
It is also possible to define $M_B$ and $M_C$ in an analogous way.
The matrix $M_B$ follows $B$ in $V_*(.A.B.C.)$ instead of $A$, and $M_C$
follows $C$.  The entire computation takes $O(n^3)$ arithmetic operations.

\begin{algorithm}
  \DontPrintSemicolon
  \LinesNumbered
  
  \KwData{Three $n \times n$ permutation matrices $A$, $B$ and $C$.}
  \KwResult{Three median candidates $M_A$, $M_B$, and $M_C$.}
  
  \BlankLine

  Compute the bases $B_1$, $B_2$, $B_3$, $B_4$, and $B_5$ with Algorithm~\ref{alg:bases}\;
  \ForEach{$i \in \{1,2,3,4,5\}$}{$P_i \leftarrow B_iB_i^t$\;}
  $M_A \leftarrow A + (B-A)P_3$\;
  $M_B \leftarrow B + (A-B)P_4$\;
  $M_C \leftarrow C + (B-C)P_2$\;

  \caption{Computation of median candidates.}
  \label{alg:med}
\end{algorithm}

We can also define a median candidate $M_I$, that follows the identity
in the subspace $V_*(.A.B.C.)$.
That is, \[M_I = A(P_1 + P_2 + P_4) + BP_3 + P_5 = A + (B-A)P_3 + (I-A)P_5\,.\]
We still have no proven results on its total
median score, but we conjecture that, for genomic matrices, $M_I$ is
better than $M_A$, $M_B$ and $M_C$, or it might even be a median. This
is due to the symmetric nature of the genomic matrices.

% Examples of our computation of the matrix median candidate are
% available in a technical report that describes a previous version of
% this procedure~\cite{Biller2012}.

\subsection{Approximation Factor}

We already know that the direct sum of the $V_*$ subspaces is $\mathbb{R}^n$, 
when $A$, $B$, and $C$ are permutation matrices. Now we can use this property 
to express the distance in terms of these subspaces dimensions. The matrix $M_A$ will have 
a total distance to $A$, $B$, and $C$ equal to:
\begin{eqnarray}
  d(M_A; A,B,C) &=& d(M,A) + d(M,B) + d(M,C) \nonumber \\
   &=& \dim V_*(.BC.A.) + \nonumber \\
   &&  \dim V_*(.A.B.C.) + \dim V_*(.AC.B.) + \nonumber \\
   &&  \dim V_*(.A.B.C.) + \dim V_*(.AB.C.) \nonumber \\
   &=& 2 \dim V_*(.A.B.C.) + \nonumber \\
   &&  \dim V_*(.AB.C.) + \dim V_*(.AC.B.) + \dim V_*(.BC.A.).
   \label{e:median_distance}
\end{eqnarray}

There are cases where $M_A$ is not a median, even if the input matrices are genomic. Take, for example:
\[A =
\begin{bmatrix}
  0 & 1 & 0 \\
  1 & 0 & 0 \\
  0 & 0 & 1
\end{bmatrix}
,
B =
\begin{bmatrix}
  1 & 0 & 0 \\
  0 & 0 & 1 \\
  0 & 1 & 0
\end{bmatrix}
\text{, and }
C =
\begin{bmatrix}
  0 & 0 & 1 \\
  0 & 1 & 0 \\
  1 & 0 & 0
\end{bmatrix}.
\]

By Equation~(\ref{e:median_distance}), the matrix $M_A$ has the following total score:
\[
  d(M_A; A, B, C) = 2 \times 2 + 0 + 0 + 0 = 4.
\]

However, the identity matrix $I$ has a better total score than $M_A$, and is actually a median in this case:
\[
  d(I; A, B, C) = 1 + 1 + 1 = 3.
\]

Thus, given that the procedure described in the Section~\ref{sec:comp} does not guarantee a matrix median, it is interesting to know whether it is an approximation algorithm, namely, whether there is a constant $\rho$ such that the candidate's total score is at most $\rho$ times the score of a median.

In general, consider permutation matrices $A$, $B$, and $C$ and a matrix $M$ such that $d(M; A, B, C)$ is minimum, that is, $M$ is a median.
There is a trivial lower bound for the median score of a matrix, easily obtained with the help of the triangle inequality, namely
\[
  d(M; A, B, C) \geq \frac{1}{2}(d(A,B) + d(B,C) + d(C,A)).
\]

According to Equation~\eqref{e:median_distance}, the median score of the approximate solution $M_A$ constructed in Section~\ref{sec:comp} is given by:
\begin{eqnarray*}
  d(M_A; A,B,C) &=& 2 \dim V_*(.A.B.C.) + \\
   &&  \dim V_*(.AB.C.) + \dim V_*(.AC.B.) + \dim V_*(.A.BC.), 
\end{eqnarray*}

For comparison, we can write the trivial lower bound in terms of 
subspace dimensions.  It suffices to write each distance as a
dimension sum of the subspaces where they differ.  The result is:
\begin{eqnarray*}
  \frac{1}{2}(d(A,B) + d(B,C) + d(C,A))
   &=& \frac{3}{2} \dim V_*(.A.B.C.) + \dim V_*(.AB.C.) + \\
   && \dim V_*(.AC.B.) + \dim V_*(.A.BC.).
\end{eqnarray*}

Then, to prove that the matrix $M_A$ is indeed an approximate solution, it suffices to show that there is a constant $\rho$ such that 
\[
  d(M_A; A, B, C) \leq \rho d(M; A, B, C),
\]
for any given matrices $A$, $B$, and $C$.

It is possible to demonstrate that $\frac{4}{3}$ is an approximate factor for our solution, as follows:
\begin{eqnarray*}
  d(M_A; A, B, C) &=& 2 \dim V_*(.A.B.C.) + \dim V_*(.AB.C.) + \dim V_*(.AC.B.) + {}\\
                  &&  {} + \dim V_*(.A.BC.)\\
                  &\leq& \frac{4}{3}[\frac{3}{2} \dim V_*(.A.B.C.) + \dim V_*(.AB.C.) + \dim V_*(.AC.B.) + {}\\
                  &&  {} + \dim V_*(.A.BC.)]\\
                  &\leq& \frac{4}{3}[\frac{1}{2}(d(A,B) + d(B,C) + d(C,A))]\\
                  &\leq& \frac{4}{3} d(M; A, B, C).\\
\end{eqnarray*}

Thus, we proved that $d(M_A; A, B, C)$ is at most $\frac{4}{3}$ times $d(M; A, B, C)$. The same result holds for $M_B$ and $M_C$.

\section{Conclusions}
\label{s:con}

We showed in this paper that it is possible to define a distance on
matrices in a way that yields exactly the algebraic distance when
restricted to permutation matrices. For the case where the input matrices
represent genomes, we have shown how to compute
matrices that approximate the median with factors at least as good as the best corner,
that is, they are approximations to the median within a factor of $\frac{4}{3}$.
In addition, we showed a construction, in the form of matrix $M_I$, that in several examples we worked out is even better than the other median candidates.  We conjecture that this matrix might be a median if $A$, $B$, and $C$ come from genomes.  The implications 
to computing algebraic genome medians can be significant.

\section{Acknowledgements}
  We thank Luiz Antonio Barrera San Martin, who suggested
  linear representations in a discussion on permutations,
  the Research Funding Agency of the State of Sao Paulo (FAPESP), for
  grants 2012/13865-7 and 2012/14104-0, and the anonymous referees for
  helpful comments that improved the paper significantly.
\bibliographystyle{splncs03}
\bibliography{median-small}

\end{document}